\documentclass[twocolumn,notitlepage,prl,superscriptaddress,longtable,longbibliography]{revtex4-2}
\usepackage{mhchem}
\usepackage{amsthm}
\usepackage{amsmath}
\usepackage{amssymb}
\usepackage{mathdots}
\usepackage{graphicx}
\usepackage{mathrsfs}
\usepackage{longtable}
\usepackage{multirow}
\usepackage{babel}
\usepackage{amsthm} 
\usepackage{xcolor}
\newtheorem{theorem}{Theorem} 
\makeatletter
 \definecolor{BLACK}{gray}{0}
 \definecolor{WHITE}{gray}{1}
 \definecolor{RED}{rgb}{1,0,0}
 \definecolor{GREEN}{rgb}{0,1,0}
 \definecolor{BLUE}{rgb}{0,0,1}
 \definecolor{CYAN}{cmyk}{1,0,0,0}
 \definecolor{MAGENTA}{cmyk}{0,1,0,0}
 \definecolor{YELLOW}{cmyk}{0,0,1,0}

\usepackage{braket}
\usepackage[backref=true,
bookmarksnumbered=true,
bookmarks=true,
bookmarksopen=true,
colorlinks=true,
citecolor=blue,
linkcolor=blue,
anchorcolor=green,
urlcolor=blue,unicode=false]{hyperref}

\let\baraccent=\= 
\renewcommand{\=}[1]{\stackrel{#1}{=}} 

\begin{document}
\title{Non-abelian Geometric Quantum Energy Pump}
\author{Yang Peng}\email{yang.peng@csun.edu}
\affiliation{Department of Physics and Astronomy, California State University, Northridge, Northridge, California 91330, USA}
\affiliation{Institute of Quantum Information and Matter and Department of Physics, California Institute of Technology, Pasadena, CA 91125, USA}

\begin{abstract}
We introduce a non-abelian geometric quantum energy pump realized by a transitionless geometric quantum drive--a time-dependent Hamiltonian supplemented by a counterdiabatic term generated by a prescribed trajectory on a smooth control manifold--that coherently transports states within a degenerate subspace. When the coordinates of the trajectory are independently addressable by external drives, the net energy transferred between drives is set by the non-abelian Berry-curvature tensor. The trajectory-averaged pumping power is separately controlled by the initial state and by the Hamiltonian topology through the Euler class. We outline an implementation with artificial atoms, which are realizable on various platforms including trapped atoms/ions, superconducting circuits, and semiconductor quantum dots. The resulting energy pump can serve as a quantum transducer or charger, and as a metrological tool for measuring phase coherences in quantum states.
\end{abstract}

\maketitle
\emph{Introduction.}---
Quantum energy pump, a quantum machine that controls the form and flow of energy among interconnected quantum systems, is one of the most important building block in quantum technology. It can function as a quantum transducer~\cite{Quantum_trans_rev}, which converts and transports energy between devices operating at different energy scales~\cite{bruzewicz2019trapped,kjaergaard2020superconducting,RevModPhys.95.025003} to create a quantum network~\cite{kimble2008quantum,caleffi2025quantum}.
In a different setting, the energy pump behaves as a charger~\cite{Harald2026} that is able to transfer energy into a quantum battery~\cite{Quantum_bat_rev}.

Two features are particularly desirable for an ideal quantum energy pump: \emph{controllability} and \emph{sustainability}. Controllability means that the pump can process energy at different frequencies and that the pumping rate can be tuned continuously. Sustainability means that the pump can operate for long times without requiring periodic resets. An important step in this direction was made in Ref.~\cite{Martin2017}, which proposed a qubit/spin driven by two external tones at different frequencies and demonstrated a chiral energy flow from one drive to the other. In the adiabatic regime---when the drive frequencies are much smaller than the instantaneous ground state energy gap---the pumping rate becomes quantized and is determined by a Chern number. While subsequent works~\cite{Peng2018_quasiperiodic,Nathan2019,Crowley2020,Qi2021} extended and generalized this framework, the reliance on adiabaticity and the associated quantization restrict the achievable tunability of the pumping power. More recently, Ref.~\cite{Psaroudaki2023} showed that energy conversion can be realized beyond the adiabatic regime at a non-quantized rate by preparing the system in a Floquet quasienergy eigenstate. However, in that approach the pumping rate does not admit a simple relation to microscopic parameters, which limits its practical use as a tunable device.

\begin{figure}[t]
    \centering
    \includegraphics[width=\linewidth]{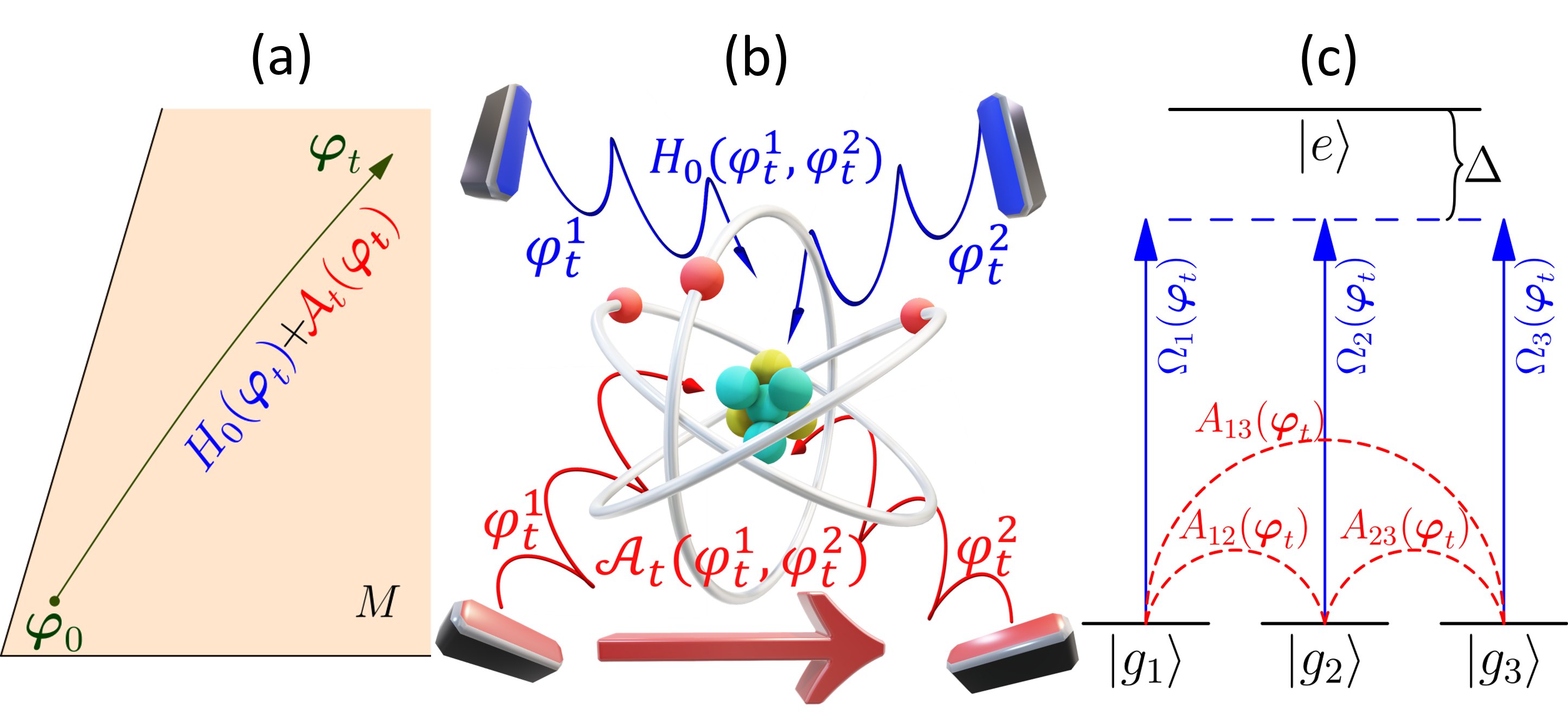}
    \caption{(a) Transitionless geometric quantum drive: a trajectory (green line) $\boldsymbol{\varphi}_t$ starting from $\boldsymbol{\varphi}_0$ on a smooth manifold $M$, and a time-dependent Hamiltonian $H_0(\boldsymbol{\varphi}_t) + \mathcal{A}_t(\boldsymbol{\varphi}_t)$ with $\mathcal{A}_t$ the Kato gauge potential of $H_0$.
    (b) An artificial atom is controlled by $H_0(\boldsymbol{\varphi}_t) + \mathcal{A}_t(\boldsymbol{\varphi}_t)$, where we assumed $\boldsymbol{\varphi}_t = (\varphi_t^1, \varphi_t^2)$, in which both components can be independently controlled by different protocols, for $H_0$ and $\mathcal{A}_t$, as indicated by blue and red sprial arrows. (c) Tripod system is a concrete experimental setup of the setting in (b). $H_0$ consists of three degenerate levels ${\ket{g_{1,2,3}}}$ which are coupled (blue arrows) to the excited state $\ket{e}$ at detuning $\Delta$ with,  $\boldsymbol{\varphi}_t$-dependent coupling strength $\Omega_{1,2,3}$. At large detuning, $\mathcal{A}_t$ consists of couplings between states in the degenerate subspace, as indicated by the red dashed lines.}
    \label{fig:fig1}
\end{figure}

A common feature of these two approaches is that the pump's state returns to itself after the pumping operation: it remains in the instantaneous ground state in the adiabatic protocol, or in a Floquet quasienergy eigenstate after each period in the Floquet protocol. This recurrence is closely connected to sustainability, as it allows repeated operation without explicit reset. By contrast, in the quantum-battery literature the sustainability of the charger is often not a primary focus: most works emphasize protocols that maximize the stored energy in minimal time~\cite{Mazzoncini2023,rodriguez2024optimal}, and the charging dynamics is typically non-steady, with the final state of the charger left unconstrained.

To achieve greater controllability, in this Letter we exploit the \emph{non-abelian} geometric structure~\cite{Wilczek} associated with a degenerate subspace, rather than restricting the dynamics to a single eigenstate as in the approaches discussed above.
To ensure sustainability, we implement a counterdiabatic (CD) protocol~\cite{demirplak2003adiabatic,demirplak2008consistency,berry2009} that constrains the evolution to remain within the degenerate manifold. Concretely, the pump is driven by a transitionless geometric protocol [Fig.~\ref{fig:fig1}(a)] that keeps the system inside the chosen degenerate subspace at arbitrary driving speed and for arbitrarily long durations. The non-abelian evolution within this subspace provides an additional control knob: the pumping rate can be tuned continuously by preparing different initial superpositions within the degenerate manifold. Finally, we show that the required ingredients can be implemented in an artificial atom with a ``tripod'' level structure, as illustrated in Figs.~\ref{fig:fig1}(b,c). Tripod-like systems have been realized experimentally using trapped atoms/ions~\cite{Han2008,Juzeli2008,Feng2020,zhao2022quantum}, superconducting circuits~\cite{faoro2003non,kamleitner2011geometric,abdumalikov2013experimental,Kumar2016,vepsalainen2019superadiabatic,Dogra2022,Setiawan2021}, and semiconductor quantum dots~\cite{Houel2014,Carter2021,zhou2024quantum}.

To allow readers to assess whether the proposed scheme is relevant in a particular scenario, we compare the different energy pumping schemes in Table \ref{tab:comparison}. 
\begin{table*}
    \centering
    \begin{ruledtabular}
    \begin{tabular}{ccccc}
    &This work  &Ref.~\cite{Martin2017,Peng2018_quasiperiodic,Nathan2019,Crowley2020} &Ref.~\cite{Psaroudaki2023}
    &Ref.~\cite{Mazzoncini2023,rodriguez2024optimal}\\
    \colrule
         Speed &finite speed  &adiabatic (small frequency)  &finite speed &finite speed  \\
         Controllability &Euler class and phase coherence  &Chern number (and system size in Ref.~\cite{Peng2018_quasiperiodic})  &no  &no  \\
         Sustainability  &yes &yes &yes &no \\
    \end{tabular}
    \end{ruledtabular}
    \caption{Comparison between different energy pumping schemes}
    \label{tab:comparison}
\end{table*}

\emph{Transitionless geometric quantum drive.}---
Consider a smooth Hamiltonian $H_{0}(\boldsymbol{\varphi})$ defined
on a $d$-dimensional manifold $M$, with local coordinate $\boldsymbol{\varphi}=(\varphi^1, \varphi^2, \dots, \varphi^d)$.
Let $\boldsymbol{\varphi}_{t}$ be a smooth trajectory on $M$ parametrized by $t$ (viewed as time), as illustrated in Fig.~\ref{fig:fig1} (a).
The triple $(H_0, M, \boldsymbol{\varphi})$ defines a geometric quantum drive~\cite{wu_geometric_2025}.

A transitionless geometric quantum drive can be achieved if we evolve an eigenstate of $H_0(\boldsymbol{\varphi}_0)$
under a new Hamiltonian $H(\boldsymbol{\varphi}_t) = H_0(\boldsymbol{\varphi}_t) + \mathcal{A}_t(\boldsymbol{\varphi}_t)$,
where $\mathcal{A}_t$ is the CD term~\cite{demirplak2003adiabatic, demirplak2008consistency,  berry2009} that ensures the time-evolved state remains to be an instantaneous eigenstate of $H_0(\boldsymbol{\varphi}_t)$.
The CD term $\mathcal{A}_t$ is not unique and one common choice is the Kato gauge potential (KGP)~\cite{kato1950adiabatic,budich2013adiabatic, Schindler2025, Ducan2025, SM} $\mathcal{A}_t = \sum_\mu\dot{\varphi}_t^\mu \mathcal{A}_\mu$,
with
\begin{equation}
    \mathcal{A}_\mu(\boldsymbol{\varphi}_t) = \frac{i}{2}\sum_n[\partial_\mu \Pi^n_{\boldsymbol{\varphi}_t},  \Pi^n_{\boldsymbol{\varphi}_t}],\label{eq:KGP}
\end{equation}
where the index $\mu$ is the directional index in parameter space ($\mu = 1,2,\dots,d$) with $\partial_\mu \equiv \partial/\partial\varphi^\mu$.
We denote the eigenspace projector of the Hamiltonian $H_0(\boldsymbol{\varphi})$ at energy $\epsilon_n(\boldsymbol{\varphi})$ as
$\Pi_{\boldsymbol{\varphi}}^{n} = \sum_\alpha \ket{n\alpha(\boldsymbol{\varphi})}\bra{n\alpha(\boldsymbol{\varphi})}$, where $\{\ket{n\alpha(\boldsymbol{\varphi})}\}$ are the corresponding (possibly) degenerate (indexed by $\alpha$) eigenstates. Notably, the KGP is a unique gauge choice that reproduces the phase evolution of eigenstates in an adiabatic evolution exactly.
The instantaneous eigenstates are parallel transported along the trajectory $\boldsymbol{\varphi}_t$ according to~\cite{SM}
\begin{equation}
\ket{n\alpha(\boldsymbol{\varphi}_t)} = \mathcal{W}(\boldsymbol{\varphi}_t \leftarrow \boldsymbol{\varphi}_0)\ket{n\alpha(\boldsymbol{\varphi}_0)},
\end{equation} where $\mathcal{W}(\boldsymbol{\varphi}_{t}\leftarrow\boldsymbol{\varphi}_{0})=\mathcal{T}\exp(-i\int_{0}^{t}dt'\,\mathcal{A}_{t}(\boldsymbol{\varphi}_{t'}))$ is the Wilson line operator ($\mathcal{T}$ is referred to as the time ordering operator).

\emph{Energy pumping}.---
Let us view $H_0(\boldsymbol{\varphi})$ as an intrinsic time-dependent system, such as a driven artificial atom, as depicted in Fig.~\ref{fig:fig1} (b).
Further, we add additional drives (viewed as battery or other devices connected to the pump) to create the KGP $\mathcal{A}_t(\boldsymbol{\varphi}_t)$, in which each component $\varphi_t^\mu$ can be controlled independently by the $\mu$th driving protocol. We initialize the system in an eigenstate of a degenerate subspace of  $H_0(\boldsymbol{\varphi}_0)$, namely choosing the initial state $\ket{\psi(0)} = \Pi_{\boldsymbol{\varphi}_{0}}^{n}\ket{\psi(0)}$. We consider the energy (integrated power) pumped~\cite{Martin2017,Peng2018_quasiperiodic, Nathan2019,Crowley2020,Qi2021,Psaroudaki2023} into the $\mu$th drive (of the KGP)
\begin{equation}
E_\mu(\boldsymbol{\varphi}_t) = \int_0^t ds\, \dot{\varphi}_s^\mu \bra{\psi(s)}\partial_\mu \mathcal{A}_t(\boldsymbol{\varphi}_s)\ket{\psi(s)} \label{eq:pumping}
\end{equation}
where $\ket{\psi(s)}$ is the time evolved state at time $s$ under $H(\boldsymbol{\varphi}_t)$.

Since the driving is transitionless, we have
\begin{equation}
\ket{\psi(t)}=e^{-i\int_{0}^{t}ds\,\epsilon_{n}(\boldsymbol{\varphi}_s)}\mathcal{W}(\boldsymbol{\varphi}_{t}\leftarrow\boldsymbol{\varphi}_{0})\ket{\psi(0)},
\end{equation}
which enables us to obtain
$E_{\mu}(\boldsymbol{\varphi}_{t})=\int_{0}^{t}ds\,\dot{\varphi}_{s}^{\mu} F_{t\mu}^{\psi}(\boldsymbol{\varphi}_s)$,
where $F^\psi_{t\mu}(\boldsymbol{\varphi}_t) = \bra{\psi(t)}\mathcal{F}^n_{t\mu}(\boldsymbol{\varphi}_t)\ket{\psi(t)}$ is the expectation value of the non-abelian Berry curvature tensor~\cite{Wilczek,SM} on the $n$th eigensubspace, which is defined as
\begin{equation}
\mathcal{F}^n_{\mu\nu}(\boldsymbol{\varphi}_t) = i[\partial_\mu \Pi^n_{\boldsymbol{\varphi}_t}, \partial_\nu \Pi^n_{\boldsymbol{\varphi}_t}].
\end{equation}
With an initial basis set $\{\ket{n\alpha(\boldsymbol{\varphi}_0)}\}$,
one can parametrize $\ket{\psi(0)} = \sum_\alpha c_\alpha \ket{n\alpha(\boldsymbol{\varphi}_0)}$, then the pumped energy can also be written as~\cite{SM}
\begin{equation}
E_\mu (\boldsymbol{\varphi}_t) = \sum_{\nu\neq\mu}\sum_{\alpha \beta}c^*_\alpha c_\beta\int_0^t ds\, \dot{\varphi}^\nu_s\dot{\varphi}^\mu_s F^{\alpha\beta}_{\nu\mu}(\boldsymbol{\varphi}_s),
\label{eq:pumping_curvature}
\end{equation}
where only index $\nu$ is summed, and $F^{\alpha\beta}_{\nu\mu} = \bra{n\alpha}\mathcal{F}^n_{\nu\mu}\ket{n\beta}$.

For simplicity, we shall choose $M$ to be an orientable closed manifold, say a $d$-torus $\mathbb{T}^d$, and parametrize
\begin{equation}
\varphi^\mu_t = \varphi^\mu_0+ f^\mu(t)
\label{eq:parametrization}
\end{equation}
with identification of $\varphi^\mu_t =\varphi^\mu_{t+2\pi}$.
Here, we assume that the function $f^\mu(t)$'s time derivative, $\dot{f}^\mu(t)$,  has a slow time dependence can be regarded as the driving frequencies that can be adjusted in time.

Note that the driving trajectory $\boldsymbol{\varphi}_t$ depends on the initial phases $\varphi_0^\mu$, which specify the starting point of the trajectory on the manifold. In many experiments these initial phases are not directly controlled---for instance, when multiple drives are free-running continuous waves and the experimental sequence begins at an arbitrary time.  In such situations, the most relevant quantity is therefore the \emph{phase-averaged} pumped energy,
$\overline{E}_\mu(t)$, defined as the energy transferred into channel $\mu$ up to time $t$ averaged over the unknown initial phases. Experimentally, this average can be obtained either by repeating the protocol many times on a single device, or by measuring an ensemble of devices that effectively sample different initial phases. The phase-averaged pumped energy can be written as
    $\overline{E}_\mu (t) = \sum_{\nu\neq \mu} \int_0^tds\,\overline{P}_{\nu\mu}(s)$,
where $\overline{P}_{\nu\mu}(t)$ is the instantaneous pumping power from drive $\nu$ to drive $\mu$, averaged over initial phases.
We find
\begin{equation}
   \overline{P}_{\nu\mu}(t) = \sum_{\alpha\beta}  c^*_\alpha c_\beta \overline{F}^{\alpha\beta}_{\nu\mu} \dot{f}^\nu(t)\dot{f}^\mu(t).
   \label{eq:pumpingpower}
\end{equation}
with the phase-averaged non-abelian Berry curvature $\overline{F}^{\alpha\beta}_{\mu\nu} = \int d\varphi^\mu d\varphi^\nu\,F^{\alpha\beta}_{\mu\nu}/(4\pi^2)$, where the integral is taken over the $(\varphi^\mu,\varphi^\nu)$ submanifold.
Eq.~\eqref{eq:pumpingpower} is one of the central results of this work; below we specialize it to a ``tripod'' implementation, where it reduces to a particularly transparent and experimentally useful expression in Eq.~\eqref{eq:central}.

Before introducing the experimental proposition, we comment on the reason why we choose the KGP $\mathcal{A}_t(\boldsymbol{\varphi}_t)$, not the entire $H(\boldsymbol{\varphi}_t)$ as driving term for the energy pump. This is because the term $H_0(\boldsymbol{\varphi}_t)$, which is the other contribution in $H(\boldsymbol{\varphi}_t)$, does not produce any pumping power when the phase average is conducted. To see this, we can write down the energy pumped due to $H_0$ as
\begin{equation}
E'_\mu(\boldsymbol{\varphi}_t) = \int_0^t ds\, \dot{\varphi}_s^\mu \bra{\psi(s)}\partial_\mu H_0(\boldsymbol{\varphi}_s)\ket{\psi(s)} \label{eq:H0pumping}.
\end{equation}
Since $\ket{\psi(s)}$ is an instantaneous eigenstate, we can apply the Hellmann–Feynman theorem, and write the integrand as $\dot{\varphi}_s^\mu  \partial_\mu\epsilon_n(\boldsymbol{\varphi}_t)$. This term averages to zero over random initial phases if one takes the parametrization in Eq.~(\ref{eq:parametrization}).

\emph{Experimental Proposition.}---A convenient physical implementation uses three long-lived states {$\ket{g_i}$} ($i=1,2,3$) coupled to a common excited state $\ket{e}$ (``tripod'' system) with a common one-photon detuning $\Delta$ and time-dependent Rabi rates $\Omega_i(t)\in \mathbb{R}$ for each state $\ket{g_i}$, which can be controlled independently. This is described by the following time-dependent Hamiltonian $H_0(\boldsymbol{\varphi}_t) = \Delta \ket{e}\bra{e} + V(\boldsymbol{\varphi}_t) + V^\dagger(\boldsymbol{\varphi}_t)$, where $V(\boldsymbol{\varphi}) =  \sum_{i=1}^3 \Omega_i(\boldsymbol{\varphi}) \ket{e}\bra{g_i}$. Such a setup can be realized experimentally in various physical platforms~\cite{Han2008,Juzeli2008,Feng2020,zhao2022quantum,faoro2003non,kamleitner2011geometric,abdumalikov2013experimental, Kumar2016, vepsalainen2019superadiabatic, Setiawan2021, Dogra2022,Houel2014,Carter2021,zhou2024quantum}.

This Hamiltonian has a doubly degenerate dark state subspace at zero energy, denoted as $\ker{H_0}$,  which is orthogonal to the excited state $\ket{e}$ and the superposition state $\ket{\tilde{g}} = \sum_j \Omega_j \ket{g_i}/\Omega$, where $\Omega = \sqrt{\sum_i \Omega_i^2}$. $H_0$ also admits two non-degenerate states $\ket{\psi_\pm}$  at finite energies $\epsilon_{\pm} = (\Delta \pm \sqrt{\Delta^2 + 4\Omega^2})/2$. These two finite energy eigenstates are obtained by superpositions of $\ket{e}$ and $\ket{\tilde{g}}$.

Let us initialize the system in a dark state $\ket{\psi(0)}\in \ker{H_0(\boldsymbol{\varphi}_0)}$, which can be done using the technique coherent population transfer~\cite{Bergmann1998, Vitanov2017}. To ensure the time-evolving state $\ket{\psi(t)}$ stays within the dark state subspace, we need to add the CD term $\mathcal{A}_t$, i.e. the KGP introduced previously.
Actually, we only need to construct the KGP projected onto the subspace spanned by the three long-lived states $\{\ket{g_{i}}\}$. To see this, assuming $\ket{\psi(t)}\in \ker{H_0(\boldsymbol{\varphi}_t)}$, then consider infinitesimal evolution
\begin{align}
\ket{\psi(t+d t)}&=\ket{\psi(t)}-i dt (H_0(\boldsymbol{\varphi}_t)+\mathcal{A}_t(\boldsymbol{\varphi}_t))\ket{\psi(t)} \nonumber \\
 &= \ket{\psi(t)} - idt\mathcal{A}_t(\boldsymbol{\varphi}_t)\ket{\psi(t)},
\end{align}
which is orthogonal to $\ket{e}$.
Thus, in order to maintain $\ket{\psi(t+dt)}\in\ker{H_0(\boldsymbol{\varphi}_{t+dt})}$, we just need to ensure $\braket{\psi(t+dt)|\tilde{g}(t+dt)} = 0$, which can be done by choosing $\mathcal{A}_t = i[\dot{\Pi}_{\boldsymbol{\varphi}_t}, \Pi_{\boldsymbol{\varphi}_t}]$, with the eigenstate projector $\Pi_{\boldsymbol{\varphi}_t} = \ket{\tilde{g}(\boldsymbol{\varphi}_t)}\bra{\tilde{g}(\boldsymbol{\varphi}_t)}$. By induction, $\ket{\psi(t)}$ stays to be a dark state.

Explicitly, the KGP $\mathcal{A}_t = i\sum_{jk}A_{jk}\ket{g_j}\bra{g_k}$, where~\cite{SM}
\begin{equation}
   A_{jk}(\boldsymbol{\varphi}_t) = \frac{1}{\Omega(\boldsymbol{\varphi}_t)^2}\left[\dot{\Omega}_{j}(\boldsymbol{\varphi}_t)\Omega_{k}(\boldsymbol{\varphi}_t)-\dot{\Omega}_{k}(\boldsymbol{\varphi}_t)\Omega_{j}(\boldsymbol{\varphi}_t)\right].
\end{equation}
These terms can be implemented with the two-photon Raman processes~\cite{Chen2010shortcut,Dalibard2011_RMP_GaugeFields}
which has been realized experimentally~\cite{Kumar2016, vepsalainen2019superadiabatic, Dogra2022}

Since the Hamiltonian $H_0$ is real symmetric, we can choose two 3D real vectors $\boldsymbol{u_1}\equiv\ket{u_1}$, and $\boldsymbol{u_2}\equiv\ket{u_2}$ as the two instantaneous basis states of $\ker{H_0}$, at each $\boldsymbol{\varphi}$, such that $\ket{\tilde{g}}\equiv \tilde{\boldsymbol{g}} = \boldsymbol{u_1}\times\boldsymbol{u_2}$, where $\times$ denotes the cross product of vectors in $\mathbb{R}^3$. The non-abelian Berry curvature in basis $\{\ket{u_1}, \ket{u_2}\}$ has zero diagonal components, due to reality of the states.
It can have nonzero off-diagonal elements $F^{12}_{\nu\mu} = -F^{12}_{\nu\mu}$, also known as the Euler form~\cite{Unal2020,bouhon2020non}, which is given by
\begin{subequations}
\begin{align}
\mathrm{Eu}_{\nu\mu} &\equiv -iF^{12}_{\nu\mu} = \tilde{\boldsymbol{g}}\cdot\left(\partial_\nu  \tilde{\boldsymbol{g}} \times \partial_\mu\tilde{\boldsymbol{g}}\right) \\
&= \braket{\partial_\nu u_1|\partial_\mu u_2} - i\braket{\partial_\nu u_2|\partial_\mu u_1}.
\end{align}
\end{subequations}

Writing $\ket{\psi(t)} = c_1\ket{u_1(\boldsymbol{\varphi}_t)} + c_2 \ket{u_2(\boldsymbol{\varphi}_t)}$, and plug it into Eq.~(\ref{eq:pumping_curvature}), we have
\begin{equation}
    E_\mu (\boldsymbol{\varphi}_t) = -2\mathrm{Im}(c^*_1 c_2)\sum_{\nu\neq\mu}\int_0^t ds\, \dot{\varphi}^\nu_s\dot{\varphi}^\mu_s \mathrm{Eu}_{\nu\mu}(\boldsymbol{\varphi}_s).
    \label{eq:single_trajectory}
\end{equation}
Let us take the parametrization in Eq.~(\ref{eq:parametrization}), and write $c_1 = c e^{i\phi_1}$, $c_2 = \sqrt{1-c^2}e^{i\phi_2}$.
If we perform average over initial drive phases, we obtain the instantaneous pumping power from drive $\nu$ to drive $\mu$ defined in Eq.~(\ref{eq:pumpingpower}) as
\begin{equation}
    \overline{P}_{\nu\mu}(t) = \dot{f}^\nu(t)\dot{f}^\mu(t)c\sqrt{1-c^2}\sin(\delta \phi)\frac{\chi_{\nu\mu}}{\pi}
    \label{eq:central}
\end{equation}
where $\delta \phi = \phi_1 - \phi_2$. Here, averaging $\mathrm{Eu}_{\nu\mu}$ gives the Euler class on the corresponding submanifold, defined as~\cite{Unal2020,bouhon2020non}
\begin{equation}
\chi_{\nu\mu} = \int d\varphi^\nu d\varphi^\mu\,\mathrm{Eu}_{\nu\mu}/(2\pi)
\end{equation}
which is always an even integer.
Note that Eqs.~(\ref{eq:single_trajectory},\ref{eq:central}) are independent of the particular choice of basis $\{|u_1\rangle,|u_2\rangle\}$ within the degenerate subspace (see Supplemental Material~\cite{SM}); it can be viewed as the specialization of the general expression~\eqref{eq:pumpingpower} to the tripod implementation.

Eq.~(\ref{eq:central}) makes the controllability of the pump explicit: the magnitude and direction of the instantaneous pumping power are tunable through the initial superposition within the degenerate manifold. In particular, $|\overline{P}_{\nu\mu}|$ is maximized at $c=1/\sqrt{2}$ and $\delta\phi=\pm \pi/2$, and changing the sign of $\delta\phi$ reverses the pumping direction. The dependence on $\sin(\delta\phi)$ highlights the role of phase coherence, allowing one to interpret the non-abelian energy pumping as a quantum-interference effect.

\emph{Two-tone drive example.}--- As an illustration, we fix the manifold $M = \mathbb{T}^2$, with two coordinates $\boldsymbol{\varphi}_t = (\varphi_t^1, \varphi^2_t)$, with parametrization $\varphi_t^i = \varphi_0^i + \omega_i t$. We then choose $\Omega_1(\boldsymbol{\boldsymbol{\varphi_t}}) = m-\cos(\varphi^1_t)-\cos(\varphi^2_t)$, $\Omega_2(\boldsymbol{\varphi}_t) = \sin\varphi_t^1$, $\Omega_3(\boldsymbol{\varphi}_t ) = \sin\varphi_t^2$. This gives nontrivial Euler class $\chi_{21} = -\chi_{12} = \pm2$, when $m$ is between $0$ and $\pm 2$.

\begin{figure*}[t]
    \centering
    \includegraphics[width=0.95\textwidth]{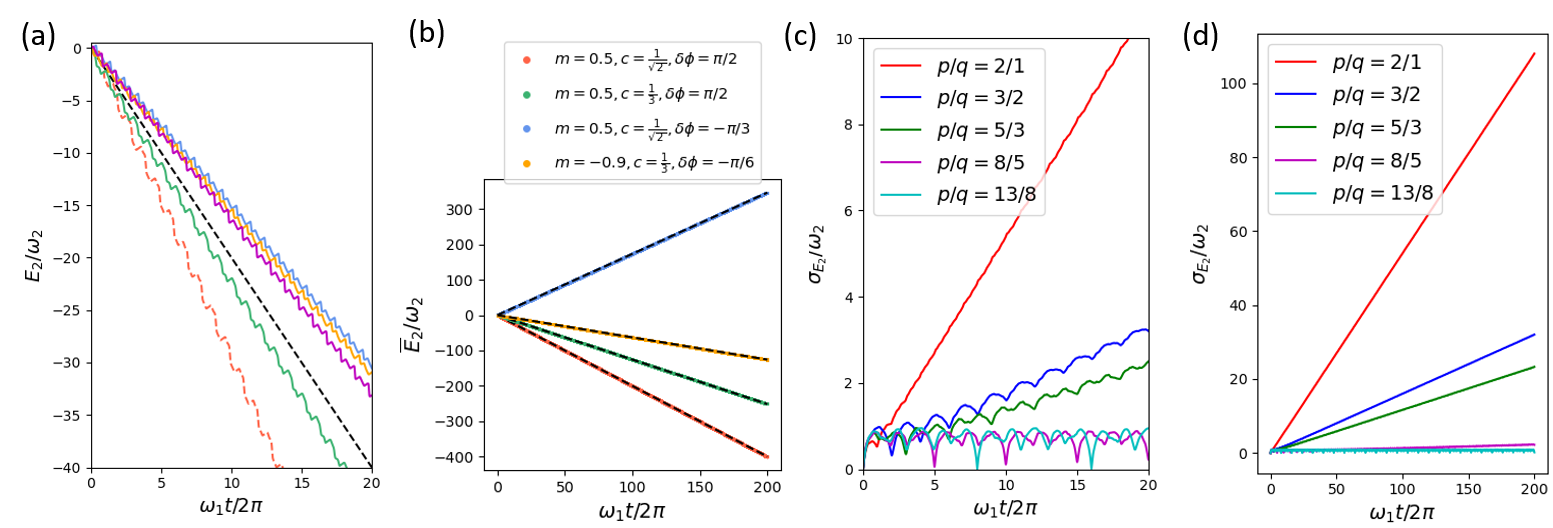}
    \caption{(a) Pumped energy $E_2(\boldsymbol{\varphi}_t)$ into drive 2 of the two-tone drive example. Five trajectories at different random initial phases are shown in colors. The black dashed line is phase-averaged result based on Eq.~(\ref{eq:central}). The parameters used are $m = 0.5, c = 1/\sqrt{2}, \delta\phi = \pi/2, p/q = 3/2$.  (b) Pumped energy $\overline{E}_2(t)$ into drive 2 averaged over 400 random initial phases sampled uniformly from $0$ to $2\pi$. The ratio is fixed at $p/q=3/2$. The different sets of other parameters that control the energy pumping power are indicated in different color. The black dashed line is the analytical result.  (c,d) Short- and long-time behaviors of standard deviation $\sigma_{E_2}$ of $E_2$ at different $p/q$ ratios for the 400 trajectories from different initializations. We fix $c=1/\sqrt{2}$, $\delta\phi = \pi/2$, $m = 0.5$. The common parameters in all figures are  $\omega = 0.4$, $\Delta = 1$.}
    \label{fig:fig2}
\end{figure*}

When $\omega_1/\omega_2$ is a irrational number, the time-dependence becomes quasiperiodic,  whereas if $\omega_1/\omega_2 = q/p$, with integers $p,q$, the system is time-periodic. Experimentally, the ratio of the two frequencies is always a rational number and we thus write $\omega_1 = \omega$, and $\omega_2 =(p/q) \omega $ in the following. The quasiperiodic case can be obtained as we can approximate any irrational number as a integer ratio. For example, the golden ratio $\frac{\sqrt{5}+1}{2}\simeq \frac{F_{N+1}}{F_N}$ where $F_N$ is the $N$th Fibonacci number~\cite{Harald2025prl}. This approximation gets better as $N$ increases.

To test our analytical result in Eq.~(\ref{eq:central}), we numerically evolve an initial state prepared in the ground state manifold of $H_0(\boldsymbol{\varphi}_0)$, with a random chosen initial phases $(\varphi^1_0, \varphi^2_0)$, and compute the pumped energy for drive $2$,  $E_2(\boldsymbol{\varphi}_t) = -E_1(\boldsymbol{\varphi}_t)$, according to Eq.~(\ref{eq:pumping}).
In Fig.~\ref{fig:fig2} (a), we show $E_2(\boldsymbol{\varphi}_t)$ in colors, for five different trajectories initialized at random phases for parameters $m = 1, c = 1/\sqrt{2}, \delta\phi = \pi/2, p/q = 3/2$.  We see that the magnitude of pumped energy overall grows linearly in time. On top of that, small oscillations are also visible. Since $p/q$ is far away from a irrational number, the slope of the linear trends deviates from the one for the phase-averaged rate $\overline{E}_2(t)/\omega_2 = -\omega_1 t/\pi$, indicated by the black lines.

In Fig.~\ref{fig:fig2} (b), we show the pumped energy $\overline{E}_2(t)$ at different set of parameters, averaged over 400 trajectories at different random initial phases uniformly sampled between $0$ and $2\pi$. We see a robust energy pumping at steady rates, coincide with the  depicted in black dashed lines, obtained from the phase-averaged pumping power calculated analytically in Eq.~(\ref{eq:central}). Notably, the direction and magnitude of the pumping power can be controlled by either tuning the Hamiltonian parameter $m$, or choosing the initial state parameters $c$ and $\delta\phi$.

To quantify how does the pumped energy deviate from its phase-averaged value, we introduce the standard deviation of $E_2$,
\begin{equation}
    \sigma_{E_2}(t) = \sqrt{\sum_{i=1}^N\frac{(E_2(\boldsymbol{\varphi}^{(i)}_t) - \overline{E}_2(t))^2}{N}},
\end{equation}
where we consider $N$ trajectories $\{\boldsymbol{\varphi}^{(i)}_t\}$ indexed by $i$ that give different pumped energy $E_2(\boldsymbol{\varphi}^{(i)}_t)$. In Figs.~\ref{fig:fig2} (c,d), we show the results at short and long time scale with $N=400$ at different $p/q$ ratios, generated by ratios of neighboring Fibonacci numbers.
We find that $\sigma_{E_2}$ overall increases linearly in time, with additional small oscillations.
The slope of this linear trend, due to Eq.~(\ref{eq:single_trajectory}), is proportional to the difference between the Euler form averaged over subregion explored by individual trajectory $\boldsymbol{\varphi}_t^{(i)}$, and the one averaged over the entire manifold (which gives the Euler class $\chi$). Note that for rational $p/q$, the trajectory always form a closed orbit on $\mathbb{T}^2$. When $p,q$ are large integers, the trajectory is a huge orbit of length $\sim 2\pi q$ wrapping around the $\mathbb{T}^2$, and thus explore most of the region of the manifold. Hence, the deviation from the trajectory-averaged value decays as the ratio $p/q$ approches the quasiperiodic limit.

\emph{Conclusions and outlook.}---In this work we introduced a non-abelian geometric quantum energy pump based on a transitionless geometric drive that confines the dynamics to a degenerate subspace while allowing evolution at arbitrary speed. The resulting pump delivers a steady and persistent energy-transfer rate and, crucially, provides high tunability. We showed that the pumping power has a clear geometric origin: it is governed by the non-abelian Berry curvature associated with the degenerate manifold. Focusing on a tripod implementation driven by two tones at distinct frequencies, we computed the energy transferred from one drive to the other and demonstrated analytically and numerically that the pumping power factorizes into two independently controllable ingredients: a state-dependent prefactor set by the initial superposition within the degenerate subspace, and a topological prefactor set by the Euler class of the real-symmetric Hamiltonian $H_0$.

Experimentally, the energy transfer corresponds to a net photon flux and can be accessed using the sideband-spectroscopy protocol proposed in Ref.~\cite{Luneau2022}. A promising route is to realize the tripod structure in a fluxonium circuit~\cite{Setiawan2021}, which can exhibit coherence times approaching the millisecond scale~\cite{Somoroff2023}. The device can be embedded in a cavity with a single port coupled to a transmission line: the incoming modes deliver the (modulated) tripod drives, while the outgoing reflected field will be analyzed by the power spectrum analyzer. The photon flux can then be extracted from the measured imbalance of the sideband spectral weights around the relevant microwave carriers. In a typical fluxonium setting, the relevant level splittings are in the GHz range, while the modulation frequencies $\omega_{1,2}$ can be in the MHz range. This separation of scales suggests that an appreciable amount of transferred energy can accumulate over a few microseconds—long compared with $2\pi/\omega_{1,2}$ yet comfortably shorter than the coherence time. A detailed device-level model is beyond the scope of this Letter and will be presented elsewhere.

Several directions merit further study. On the theory side, the appearance of the Euler class in our two-tone protocol relies on a real two-dimensional degenerate subspace and a two-dimensional control manifold, consistent with $H^2(\mathbb{T}^2,\mathbb{Z})=\mathbb{Z}$. For higher degeneracy $r\geq 3$ on $\mathbb{T}^2$, no analogous Euler-class invariant exists because $H^{r}(\mathbb{T}^2,\mathbb{Z})=0$ for $r\geq2$. Realizing topological invariants associated with higher-rank real bundles therefore likely requires introducing additional independent drive phases (i.e., more frequencies), thereby enlarging the effective control manifold; exploring such multi-tone generalizations is an interesting open problem. On the practical side, non-abelian energy pumps could serve as tunable modules for quantum technology, e.g., as transducers or chargers. Moreover, because the pumping power is directly sensitive to phase coherence within the degenerate manifold, the same setup may provide an experimentally accessible probe of coherence in multilevel quantum states. Establishing these applications will require detailed modeling of the pump when it is coupled to external devices and environments.

\emph{Acknowledgment.}--- This work is supported by the US National Science Foundation (NSF) Grant No. PHY-2216774.
The numerical simulation is supported by NSF instrument grant  DMR-2406524.

\bibliography{manuscript}

\clearpage
\begin{widetext}
\section*{Supplemental Material}

\section{Properties of Kato gauge potential (KGP)}
Consider a smooth Hamiltonian $H_{0} (\boldsymbol{\varphi})$ with
parameters $\boldsymbol{\varphi}=(\varphi^{1},\varphi^{2},\dots)$
defined on a smooth manifold $M$. Let us denote the eigenspace projector
as $\Pi_{\boldsymbol{\varphi}}^{n}$. In the following, we provide several theorems that describe the properties of the KGP. 
\begin{theorem}
\label{thm1}Define the KGP along the direction $\varphi_{\mu}$
as 
\begin{equation}
\mathcal{A}_{\mu}(\boldsymbol{\varphi})=\frac{i}{2}\sum_{n}[\partial_{\mu}\Pi_{\boldsymbol{\varphi}}^{n},\Pi_{\boldsymbol{\varphi}}^{n}]
\end{equation}
 where $\partial_{\mu}\equiv\partial/\partial\varphi^{\mu}$. We have
\begin{equation}
\Pi_{\boldsymbol{\varphi}}^{n}\mathcal{A}_{\mu}(\boldsymbol{\varphi})\Pi_{\boldsymbol{\varphi}}^{n}=0
\end{equation}
\end{theorem}

Consider a smooth curve on $M$ denoted as $\boldsymbol{\varphi}_{t}$
parameterized by $t$ (viewed as time). Let us initialize the system
in an eigenstate $\ket{\psi(0)}\equiv\Pi_{\boldsymbol{\varphi}_{0}}^{n}\ket{\psi(0)}$
of $H(\boldsymbol{\varphi}_{t=0})$.  
\begin{theorem}
\label{thm2}
Let us construct the time-dependent Hamiltonian $H(\boldsymbol{\varphi}_t)  = H_0(\boldsymbol{\varphi}_t) + \mathcal{A}_t(\boldsymbol{\varphi}_t)$, and
consider the Schr\"odinger equation 
\begin{equation}
i\partial_{t}\ket{\psi(t)}=H(\boldsymbol{\varphi}_t)\ket{\psi(t)},
\end{equation}
under the initial condition $\ket{\psi(0)}=\Pi_{\boldsymbol{\varphi}_0}^n\ket{\psi(0)}$, 
then $\Pi_{\boldsymbol{\varphi}_{t}}^{n}\ket{\psi(t)}=\ket{\psi(t)}$
at any time $t$. Thus, $\mathcal{A_t}$ is the counterdiabatic term.
\end{theorem}
\begin{proof}
We can consider an infinitesimal $\delta t\to0$ time evolution
(neglecting terms $o(\delta t^{2})$)
\begin{equation}
\ket{\psi(t+\delta t)}=\ket{\psi(t)}-i\delta t\left[H_{0}(\boldsymbol{\varphi}_{t})+\mathcal{A}_{t}(\boldsymbol{\varphi}_{t})\right]\ket{\psi(t)},
\end{equation}
and compute
\begin{align}
\Pi_{\boldsymbol{\varphi}_{t+\delta t}}^{n}\ket{\psi(t+\delta t)} & =\left[\Pi_{\boldsymbol{\varphi}_{t}}^{n}+\delta t\partial_{t}\Pi_{\boldsymbol{\varphi}_{t}}^{n}\right]\left[1-i\delta tH_{0}(\boldsymbol{\varphi}_{t})-i\delta t\mathcal{A}_{t}(\boldsymbol{\varphi}_{t})\right]\ket{\psi(t)}\\
 & =\left[\Pi_{\boldsymbol{\varphi}_{t}}^{n}-\delta t\left(i\Pi_{\boldsymbol{\varphi}_{t}}^{n}H_{0}(\boldsymbol{\varphi}_{t})+i\Pi_{\boldsymbol{\varphi}_{t}}^{n}\mathcal{A}_{t}(\boldsymbol{\varphi}_{t})-\partial_{t}\Pi_{\boldsymbol{\varphi}_{t}}^{n}\right)\right]\ket{\psi(t)}.
\end{align}
Note that $[H_{0}(\boldsymbol{\varphi}_{t}),\Pi_{\boldsymbol{\varphi}_{t}}^{n}]=0$
and $\mathcal{A}_{t}(\boldsymbol{\varphi}_{t})$ has no matrix elements
between states at the same energy, namely $\Pi_{\boldsymbol{\varphi}_{t}}^{n}\mathcal{A}_{t}(\varphi_{t})\Pi_{\boldsymbol{\varphi}_{t}}^{n}=0$, by Theorem~\ref{thm1}.
Assuming $\Pi_{\boldsymbol{\varphi}_{t}}^{n}\ket{\psi(t)}=\ket{\psi(t)}$,
we have
\begin{align*}
\mathcal{A}_{t}(\varphi_{t})\ket{\psi(t)} & =i(\partial_{t}\Pi_{\boldsymbol{\varphi}_{t}}^{n})\ket{\psi(t)}.
\end{align*}
and thus
\begin{equation}
\Pi_{\boldsymbol{\varphi}_{t+\delta t}}^{n}\ket{\psi(t+\delta t)}=\left[1-i\delta t\left(H_{0}(\boldsymbol{\varphi}_{t})+\mathcal{A}_{t}(\boldsymbol{\varphi}_{t})\right)\right]\ket{\psi(t)}=\ket{\psi(t+\delta t)}.
\end{equation}
Hence, $\Pi^n_{\boldsymbol{\varphi}_{t}}\ket{\psi(t)}=\ket{\psi(t)}$
whenever $\Pi_{\boldsymbol{\varphi}_{0}}^{n}\ket{\psi(0)}=\ket{\psi(0)}$.
\end{proof}

Because of Theorem~\ref{thm2}, if one takes $\{\ket{n\alpha}\}$ as a basis for the eigensubspace of $H_0(\boldsymbol{\varphi}_0)$ at energy $\epsilon_n(\boldsymbol{\varphi}_0)$, we can define a parallelly transported basis along the trajectory $\boldsymbol{\varphi}_t$ as $\{\ket{n\alpha(\boldsymbol{\varphi}_t)}\}$, with $\ket{n\alpha(\boldsymbol{\varphi}_0)}\equiv\ket{n\alpha}$. Explictly, we can write
\begin{equation}
    \ket{n\alpha(\boldsymbol{\varphi}_t)} = \mathcal{W}(\boldsymbol{\varphi}_t\leftarrow\boldsymbol{\varphi}_0)\ket{n\alpha(\boldsymbol{\varphi}_0)},
\end{equation}
with Wilson line operator generated by the KGP $\mathcal{A}_t$
\begin{equation}
    \mathcal{W}(\boldsymbol{\varphi}_t\leftarrow\boldsymbol{\varphi}_0) = 
    \mathcal{T}\exp(-i\int_{0}^{t}dt'\,\mathcal{A}_{t}(\boldsymbol{\varphi}_{t'})) =
    \mathcal{P}\exp(-i\int_{\boldsymbol{\varphi}_t}  d\boldsymbol{\varphi} \cdot \boldsymbol{\mathcal{A}}(\boldsymbol{\varphi})),
\end{equation}
where $\boldsymbol{\mathcal{A}} = (\mathcal{A}_1,\mathcal{A}_2,\dots,\mathcal{A}_d)$. This means that from an initial state $\ket{\psi(0)}$ in this subpsace, parametrized as
\begin{equation}
    \ket{\psi(0)} = \sum_\alpha c_\alpha \ket{n\alpha(\boldsymbol{\varphi}_0)},
\end{equation}
the time-evolved state can be written as
\begin{equation}
    \ket{\psi(t)} = e^{-i\int_0^t ds\, \epsilon_n(\boldsymbol{\varphi}_s)} \mathcal{W}(\boldsymbol{\varphi}_t\leftarrow\boldsymbol{\varphi}_0)=e^{-i\int_0^t ds\, \epsilon_n(\boldsymbol{\varphi}_s)} \sum_\alpha c_\alpha \ket{n\alpha(\boldsymbol{\varphi}_t)}.
\end{equation}
This implies that the expectation value of any operator $\mathcal{F}$, the expectation value can be computed as 
\begin{equation}
    \bra{\psi(s)}\mathcal{F}\ket{\psi(s)} = \sum_{\alpha\beta}c_\alpha^* c_\beta F^{\alpha\beta},\quad F^{\alpha\beta} = \bra{\alpha}\mathcal{F}\ket{\beta}.
\end{equation}

The following two theorems are used to derive the energy pumping formula (Eq.~(6)) of the main text. Particularly, Theorem~\ref{thm4} relates the derivative $\partial_\mu \mathcal{A}_t$ to the non-abelian Berry curvature tensor $\mathcal{F}_{t\mu}$.

\begin{theorem}
\label{thm3}
Let us construct a generic projector $\Pi_{\boldsymbol{\varphi}}=\sum_{m\in S}\Pi_{\boldsymbol{\varphi}}^{m}$
as a partial sum of indices over a set $S$. Then
\begin{equation}
\partial_{\mu}\Pi_{\boldsymbol{\varphi}}=-i[\mathcal{A}_{\mu}(\boldsymbol{\varphi}),\Pi_{\boldsymbol{\varphi}}].
\end{equation}
\end{theorem}
\begin{proof}
Let us compute the right hand side of the above equation
\begin{align}
-i[\mathcal{A}_{\mu},\Pi_{\boldsymbol{\varphi}}] & =\frac{1}{2}\sum_{n}\sum_{m\in S}\left[[\partial_{\mu}\Pi_{\boldsymbol{\varphi}}^{n},\Pi_{\boldsymbol{\varphi}}^{n}],\Pi_{\boldsymbol{\varphi}}^{m}\right]\nonumber \\
 & =\frac{1}{2}\sum_{n}\sum_{m\in S}\left((\partial_{\mu}\Pi_{\boldsymbol{\varphi}}^{n})\Pi_{\boldsymbol{\varphi}}^{n}\Pi_{\boldsymbol{\varphi}}^{m}-\Pi_{\boldsymbol{\varphi}}^{n}(\partial_{\mu}\Pi_{\boldsymbol{\varphi}}^{n})\Pi_{\boldsymbol{\varphi}}^{m}-\Pi_{\boldsymbol{\varphi}}^{m}(\partial_{\mu}\Pi_{\boldsymbol{\varphi}}^{n})\Pi_{\boldsymbol{\varphi}}^{n}+\Pi_{\boldsymbol{\varphi}}^{m}\Pi_{\boldsymbol{\varphi}}^{n}(\partial_{\mu}\Pi_{\boldsymbol{\varphi}}^{n})\right)\nonumber \\
 & =\frac{1}{2}\sum_{m\in S}\left[(\partial_{\mu}\Pi_{\boldsymbol{\varphi}}^{m})\Pi_{\boldsymbol{\varphi}}^{m}+\Pi_{\boldsymbol{\varphi}}^{m}(\partial_{\mu}\Pi_{\boldsymbol{\varphi}}^{m})\right]-\frac{1}{2}\sum_{m\in S}\sum_{n\neq m}\left[\Pi_{\boldsymbol{\varphi}}^{n}(\partial_{\mu}\Pi_{\boldsymbol{\varphi}}^{n})\Pi_{\boldsymbol{\varphi}}^{m}+\Pi_{\boldsymbol{\varphi}}^{m}(\partial_{\mu}\Pi_{\boldsymbol{\varphi}}^{n})\Pi_{\boldsymbol{\varphi}}^{n}\right]\nonumber \\
 & =\frac{1}{2}\partial_{\mu}\Pi_{\boldsymbol{\varphi}}+\frac{1}{2}\sum_{m\in S}\sum_{n\neq m}\left[\Pi_{\boldsymbol{\varphi}}^{n}(\partial_{\mu}\Pi_{\boldsymbol{\varphi}}^{m})+(\partial_{\mu}\Pi_{\boldsymbol{\varphi}}^{m})\Pi_{\boldsymbol{\varphi}}^{n}\right]\nonumber \\
 & =\frac{1}{2}\partial_{\mu}\Pi_{\boldsymbol{\varphi}}+\frac{1}{2}\sum_{m\in S}\left[(1-\Pi_{\boldsymbol{\varphi}}^{m})(\partial_{\mu}\Pi_{\boldsymbol{\varphi}}^{m})+(\partial_{\mu}\Pi_{\boldsymbol{\varphi}}^{m})(1-\Pi_{\boldsymbol{\varphi}}^{m})\right]\nonumber \\
 & =\partial_{\mu}\Pi_{\boldsymbol{\varphi}}.
\end{align}
Note that in the third line, we have used the property that $\Pi_{\boldsymbol{\varphi}}^{n}(\partial_{\mu}\Pi_{\boldsymbol{\varphi}}^{n})\Pi_{\boldsymbol{\varphi}}^{n}=0.$
\end{proof}

\begin{theorem}
\label{thm4}Define the subspace projected Kato potential 
\begin{equation}
\mathcal{A}_{\mu}^{(\Pi)}(\boldsymbol{\varphi})=\Pi_{\boldsymbol{\varphi}}\mathcal{A}_{\mu}(\boldsymbol{\varphi})\Pi_{\boldsymbol{\varphi}}.
\end{equation}
Then
\begin{equation}
\Pi_{\boldsymbol{\varphi}}(\partial_{\mu}\mathcal{A}_{\nu}(\boldsymbol{\varphi}))\Pi_{\boldsymbol{\varphi}}=\Pi_{\boldsymbol{\varphi}}(\partial_{\mu}\mathcal{A}_{\nu}^{(\Pi)}(\boldsymbol{\varphi}))\Pi_{\boldsymbol{\varphi}}-\mathcal{F}_{\mu\nu}^{(\Pi)}(\boldsymbol{\varphi})
\end{equation}
where the last term is the non-abelian Berry curvature defined on
this subspace
\begin{equation}
\mathcal{F}_{\mu\nu}^{(\Pi)}(\boldsymbol{\varphi})=i\Pi_{\boldsymbol{\varphi}}[\partial_{\mu}\Pi_{\boldsymbol{\varphi}},\partial_{\nu}\Pi_{\boldsymbol{\varphi}}]\Pi_{\boldsymbol{\varphi}}.
\label{seq:berry_curvature}
\end{equation}
\end{theorem}
\begin{proof}
Let us compute the projected derivative of the Kato gauge potential
(we drop the notation of explicit $\boldsymbol{\varphi}$-dependence
in the following for simplicity)
\begin{align*}
\Pi(\partial_{\mu}\mathcal{A}_{\nu})\Pi & =\partial_{\mu}\left(\Pi\mathcal{A}_{\nu}\Pi\right)-(\partial_{\mu}\Pi)\mathcal{A}_{\nu}\Pi-\Pi\mathcal{A}_{\nu}(\partial_{\mu}\Pi)\\
 & =\partial_{\mu}\mathcal{A}_{\nu}^{(\Pi)}+i[\mathcal{A}_{\mu},\Pi]\mathcal{A}_{\nu}\Pi+i\Pi\mathcal{A}_{\nu}[\mathcal{A}_{\mu},\Pi]\\
 & =\partial_{\mu}\mathcal{A}_{\nu}^{(\Pi)}+i\mathcal{A}_{\mu}\Pi\mathcal{A}_{\nu}\Pi-i\Pi\mathcal{A}_{\mu}\mathcal{A}_{\nu}\Pi+i\Pi\mathcal{A}_{\nu}\mathcal{A}_{\mu}\Pi-i\Pi\mathcal{A}_{\nu}\Pi\mathcal{A}_{\mu}.
\end{align*}
Applying projector $\Pi$ on from left and right on both sides of
the above equation, we obtain
\begin{equation}
\Pi(\partial_{\mu}\mathcal{A}_{\nu})\Pi=\Pi(\partial_{\mu}\mathcal{A}_{\nu}^{(\Pi)})\Pi-i\left[\Pi\mathcal{A}_{\mu}(1-\Pi)\mathcal{A}_{\nu}\Pi-\Pi\mathcal{A}_{\nu}(1-\Pi)\mathcal{A}_{\mu}\Pi\right]\label{eq:project_der}
\end{equation}
From the non-abelian Berry curvature defined in terms of the projector in Eq.~(\ref{seq:berry_curvature}),
and using Theorem \ref{thm3}, we obtain
\begin{align}
\mathcal{F}_{\mu\nu}^{(\Pi)} & =i\Pi[\partial_{\mu}\Pi,\partial_{\nu}\Pi]\Pi=-i\Pi\left[[\mathcal{A}_{\mu},\Pi],[\mathcal{A}_{\nu},\Pi]\right]\Pi\nonumber \\
 & =-i\Pi\left[\mathcal{A}_{\mu}\Pi\mathcal{A}_{\nu}\Pi-\Pi\mathcal{A}_{\mu}\mathcal{A}_{\nu}\Pi-\mathcal{A}_{\mu}\Pi\Pi\mathcal{A}_{\nu}+\Pi\mathcal{A}_{\mu}\Pi\mathcal{A}_{\nu}\right]\Pi-(\mu\leftrightarrow\nu)\nonumber \\
 & =-i(\Pi\mathcal{A}_{\mu}\Pi\mathcal{A}_{\nu}\Pi-\Pi\mathcal{A}_{\mu}\mathcal{A}_{\nu}\Pi)-(\mu\leftrightarrow\nu)\nonumber \\
 & =i\Pi\mathcal{A}_{\mu}(1-\Pi)\mathcal{A}_{\nu}\Pi-i\Pi\mathcal{A}_{\nu}(1-\Pi)\mathcal{A}_{\mu}\Pi,
\end{align}
which is exactly the second term in Eq.~(\ref{eq:project_der}).
Note that $\Pi\mathcal{F}_{\mu\nu}^{(\Pi)}\Pi=\mathcal{F}_{\mu\nu}^{(\Pi)}$.
\end{proof}

In the last theorem, we show that the non-abelian Berry curvature defined in terms of the projectors in Eq.~(\ref{seq:berry_curvature}) coincide with definition of Wilczek and Zee~\cite{Wilczek}.

\begin{theorem}
Let $\{\ket{\alpha}\}$ be a basis of a subspace with projector $\Pi$,
namely $\Pi=\sum_{\alpha}\ket{\alpha}\bra{\alpha}$. Let us denote
the matrix elements $F_{\mu\nu}^{\alpha\beta}\equiv\bra{\alpha}\mathcal{F}_{\mu\nu}^{(\Pi)}\ket{\beta}$.
Then this matrix $F_{\mu\nu}$ can be computed from 
\begin{equation}
F_{\mu\nu}=\partial_{\mu}A_{\nu}-\partial_{\nu}A_{\mu}-i[A_{\mu},A_{\nu}],
\end{equation}
where $A_{\mu}$ is the non-abelian Wilczek-Zee connection, generalizing
abelian Berry connection, defined as $A_{\mu}^{\alpha\beta}=\bra{\alpha}i\partial_{\mu}\ket{\beta}.$ 
\end{theorem}

\begin{proof}
Using the definition of the non-abelian Berry curvature, we have
\begin{align}
F_{\mu\nu}^{\alpha\beta} & =i\bra{\alpha}[\partial_{\mu}\Pi,\partial_{\nu}\Pi]\ket{\beta}\\
 & =i\bra{\partial_{\mu}\alpha}(1-\Pi)\ket{\partial_{\nu}\beta}-(\mu\leftrightarrow\nu)\\
 & =\left[i\braket{\partial_{\mu}\alpha|\partial_{\nu}\beta}-i\sum_{\gamma}\braket{\partial_{\mu}\alpha|\gamma}\braket{\gamma|\partial_{\nu}\beta}\right]-(\mu\leftrightarrow\nu)\\
 & =\partial_{\mu}A_{\nu}^{(\alpha\beta)}-\partial_{\nu}A_{\mu}^{(\alpha\beta)}-i\sum_{\gamma}\left[A_{\mu}^{\alpha\gamma}A_{\nu}^{\gamma\beta}-A_{\nu}^{\alpha\gamma}A_{\mu}^{\gamma\beta}\right]
\end{align}
\end{proof}

\section{Derivation of the energy pumping formula}
The energy pumped into drive $\mu$ can be computed via
\begin{equation}
E_\mu(\boldsymbol{\varphi}_t) = \int_0^t ds\, \dot{\varphi}_s^\mu \bra{\psi(s)}\partial_\mu \mathcal{A}_t(\boldsymbol{\varphi}_s)\ket{\psi(s)}.
\end{equation}
Since $\ket{\psi(s)} = \Pi_{\boldsymbol{\varphi}_s}^n\ket{\psi(s)}$, we have
\begin{equation}
\bra{\psi(s)}\partial_\mu \mathcal{A}_t(\boldsymbol{\varphi}_s)\ket{\psi(s)} = \bra{\psi(s)}\Pi_{\boldsymbol{\varphi}_s}^n\partial_\mu \mathcal{A}_t(\boldsymbol{\varphi}_s)\Pi_{\boldsymbol{\varphi}_s}^n\ket{\psi(s)}.
\end{equation}
Using Theorems~\ref{thm1} and \ref{thm4}, we have
\begin{equation}
    \Pi_{\boldsymbol{\varphi}_s}^n\partial_\mu \mathcal{A}_t(\boldsymbol{\varphi}_s)\Pi_{\boldsymbol{\varphi}_s}^n = -\mathcal{F}_{\mu t}^n = \mathcal{F}_{t\mu} = \sum_\nu \dot{\varphi}_t^\nu \mathcal{F}_{\nu\mu}^n.
\end{equation}

Noting that $\mathcal{F}_{\mu\mu}^n = 0$, we arrive at the energy pumping formula of the main text
\begin{equation}
    E_\mu(\boldsymbol{\varphi}_t) = \sum_{\nu\neq \mu}\int_0^t ds\, \dot{\varphi}_s^\nu \dot{\varphi}_s^\mu \bra{\psi(s)}\mathcal{F}_{\nu\mu}^n(\boldsymbol{\varphi}_s)\ket{\psi(s)}.
\end{equation}

\section{Derivation of the KGP in the tripod system}
The dark state projector $\Pi = 1 - P$, where $P = \ket{\tilde{g}}\bra{\tilde{g}}$, with
\begin{equation}
    \ket{\tilde{g}} = \sum_j \Omega_j \ket{g_i}/\Omega, \quad \Omega = \sqrt{\sum_i \Omega_i^2} 
\end{equation}

The KGP can be computed as
\begin{align}
    \mathcal{A}_t &= \frac{i}{2}\left\{[\partial_tP,P] + [\partial_t(1-P), (1-P)]\right\} \notag \\
    &= i(\dot{P}P - P\dot{P}) \notag \\
   &=\frac{i}{\Omega^{2}}\sum_{jk}(\dot{\Omega}_{j}\Omega_{k}+\Omega_{j}\dot{\Omega}_{k})\left[\ket{g_{j}}\bra{g_{k}}P-P\ket{g_{j}}\bra{g_{k}}\right] \notag \\ 
   &= \frac{i}{\Omega^{4}}\sum_{jkl}\left(\dot{\Omega}_{j}\Omega_{k}^{2}\Omega_{l}\ket{g_{j}}\bra{g_{l}}-\Omega_{l}\Omega_{j}^{2}\dot{\Omega}_{k}\ket{g_{l}}\bra{g_{k}}\right) \notag \\
   &= \frac{i}{\Omega^{2}}\sum_{jk}(\dot{\Omega}_{j}\Omega_{k}-\dot{\Omega}_{k}\Omega_{j})\ket{g_{j}}\bra{g_{k}}.
\end{align}

\section{Gauge invariance of the energy pumping formula}
In this section, we verify that the pumped energy into drive $\mu$ when the system is evolved along trajectory $\boldsymbol{\varphi}_t$
\begin{equation}
    E_\mu (\boldsymbol{\varphi}_t) = -2\mathrm{Im}(c^*_1 c_2)\sum_{\nu\neq\mu}\int_0^t ds\, \dot{\varphi}^\nu_s\dot{\varphi}^\mu_s \mathrm{Eu}_{\nu\mu}(\boldsymbol{\varphi}_s).
    \label{eq:single_trajectory}
\end{equation}
does not depend on the choice of basis $\{\ket{u_1},\ket{u_2}\}$.

Let us denote the system state as 
\begin{equation}
\ket{\psi(t)} = c_1\ket{u_1(\boldsymbol{\varphi}_t)} + c_2 \ket{u_2(\boldsymbol{\varphi}_t)}
=c'_1\ket{u'_1(\boldsymbol{\varphi}_t)} + c'_2 \ket{u'_2(\boldsymbol{\varphi}_t)},
\end{equation}
where $\{\ket{u'_1},\ket{u'_2}\}$ is a different basis related to $\{\ket{u_1},\ket{u_2}\}$ by a local $O(2)$ gauge transformation
\begin{equation}
    \begin{pmatrix}
\boldsymbol u_1'\\
\boldsymbol u_2'
    \end{pmatrix} = R
    \begin{pmatrix}
\boldsymbol u_1\\
\boldsymbol u_2
    \end{pmatrix}, \quad
    R \in O(2), \quad\boldsymbol{u}_i = \ket{u_i}.
\end{equation}
Explicitly, we can parametrize the $O(2)$ matrix 
\begin{equation}
R = 
\begin{pmatrix}
    \cos\theta & \sin\theta \\
    -\sin\theta & \cos\theta
\end{pmatrix}
\end{equation}
with $\det R = +1$ (namely SO(2)), or
\begin{equation}
 R = 
\begin{pmatrix}
    \cos\theta & \sin\theta \\
    \sin\theta & -\cos\theta
\end{pmatrix}   
\end{equation}
with $\det R = -1$. 

Note that the Euler form can also be written as
\begin{equation}
    \mathrm{Eu}_{\mu\nu}
=
\tilde{\boldsymbol g}\cdot
\Big(\partial_\mu \tilde{\boldsymbol g}\times \partial_\nu \tilde{\boldsymbol g}\Big).
\label{eq:EulerForm_gtilde}
\end{equation}
where $\tilde{\boldsymbol g}=\boldsymbol u_1 \times \boldsymbol u_2$, and $\boldsymbol{u}_i \equiv \ket{u_i}$.

For an orientation-preserving transformation ($\det R = 1$)
\begin{equation}
\begin{pmatrix}
\boldsymbol u_1'\\
\boldsymbol u_2'
\end{pmatrix}
=
\begin{pmatrix}
\cos\theta & \sin\theta\\
-\sin\theta & \cos\theta
\end{pmatrix}
\begin{pmatrix}
\boldsymbol u_1\\
\boldsymbol u_2
\end{pmatrix},
\end{equation}
we have
\begin{align}
\tilde{\boldsymbol g}'
&=\boldsymbol u_1'\times \boldsymbol u_2'
\nonumber\\
&=(\cos\theta\,\boldsymbol u_1+\sin\theta\,\boldsymbol u_2)\times
(-\sin\theta\,\boldsymbol u_1+\cos\theta\,\boldsymbol u_2)
\nonumber\\
&=(\cos^2\theta+\sin^2\theta)\,(\boldsymbol u_1\times \boldsymbol u_2)
=\tilde{\boldsymbol g}.
\label{eq:gtilde_invariant}
\end{align}
Therefore $\tilde{\boldsymbol g}$ is invariant under $SO(2)$ gauge transformations, and so is the Euler form $\mathrm{Eu}_{\mu\nu}$ defined in Eq.~\eqref{eq:EulerForm_gtilde}.

If instead the basis transformation reverses orientation (i.e. $\det R=-1$, such as exchanging $\boldsymbol u_1\leftrightarrow \boldsymbol u_2$), then
\begin{equation}
\tilde{\boldsymbol g}'=\boldsymbol u_1'\times \boldsymbol u_2' = -\,(\boldsymbol u_1\times \boldsymbol u_2)=-\tilde{\boldsymbol g},
\end{equation}
which implies
\begin{equation}
\mathrm{Eu}'_{\mu\nu} = -\,\mathrm{Eu}_{\mu\nu}.
\end{equation}
Hence, the Euler form is gauge invariant under orientation-preserving changes of basis and flips sign under orientation reversal.

Denoting
\begin{equation}
    \boldsymbol{c}=\begin{pmatrix}
        c_1 \\ c_2
    \end{pmatrix}, \quad 
    \boldsymbol{\tilde{c}}=\begin{pmatrix}
        \tilde{c}_1 \\ \tilde{c}_2
    \end{pmatrix},
\end{equation}
we have $\boldsymbol{\tilde{c}} = R \boldsymbol{c}$.

Note that the term 
\begin{equation}
    -2\mathrm{Im}(c^*_1 c_2) = 
   \begin{pmatrix}
       c_1^* & c_2^*
   \end{pmatrix} 
   \begin{pmatrix}
       0 & -1 \\
       1 & 0
   \end{pmatrix}
   \begin{pmatrix}
       c_1 \\ c_2
   \end{pmatrix}
   =\boldsymbol{c}^\dagger (-i\sigma_y)\boldsymbol{c},
\end{equation}
where $\sigma_y$ is the Pauli matrix. Thus, we have that $-2\mathrm{Im}(c^*_1 c_2)$ is invariant if $\det R = 1$ and it flips the sign if $\det R = -1$.

Hence, the term $-2\mathrm{Im}(c_1^*c_2)\mathrm{Eu}_{\nu\mu}$ is invariant under transformation $R$, which verifies the $O(2)$ gauge invariance of the energy pumping formula.

\end{widetext}
\end{document}